\documentclass[conference]{IEEEtran}
\IEEEoverridecommandlockouts
\usepackage[OT1]{fontenc} 
\usepackage{cite}
\usepackage{amsmath,amssymb,amsfonts}
\usepackage{algorithmic}
\usepackage{algorithm}
\usepackage{graphicx}
\usepackage{textcomp}
\usepackage{xcolor}
\usepackage{amsthm}

\newtheorem{thm}{Theorem}

\newtheorem{corollary}{Corollary}

\def\BibTeX{{\rm B\kern-.05em{\sc i\kern-.025em b}\kern-.08em
		T\kern-.1667em\lower.7ex\hbox{E}\kern-.125emX}}
\begin{document}
	
	\title{On Decentralized Multi-Transmitter Coded Caching}
	
	\author{
		\IEEEauthorblockN{Mohammad Mahmoudi\IEEEauthorrefmark{1}, Mohammad Javad Sojdeh\IEEEauthorrefmark{1}, Seyed Pooya Shariatpanahi\IEEEauthorrefmark{1}}
		\IEEEauthorblockA{\IEEEauthorrefmark{1}School of Electrical and Computer Engineering, College of Engineering, University of Tehran, Tehran, Iran
			\\\{m.mahmoudi13, sojdehjavad1376, p.shariatpanahi\}@ut.ac.ir}
	}

\maketitle
	
\begin{abstract}
This paper investigates a setup consisting of multiple transmitters serving multiple cache-enabled clients through a linear network, which covers both wired and wireless transmission situations. We investigate decentralized coded caching scenarios in which there is either no cooperation or limited cooperation between the clients at the cache content placement phase.
		For the fully decentralized caching case (i.e., no cooperation) we analyze the performance of the system in terms of the \emph{Coding Delay} metric. Furthermore, we investigate a hybrid cache content placement scenario in which there are two groups of users with different cache content placement situations (i.e., limited cooperation). Also, we examine the effect of finite file size in above scenarios.
	\end{abstract}
	
\section{Introduction}
In recent years, due to the unprecedented growth in demand for multimedia services, the usage of distributed memories across the network has been proposed to replicate the requested content as a promising method to alleviate network bandwidth bottlenecks, which is also known as \emph{Caching} \cite{Bastug_2014}, \cite{Golrezaei_2013}. Specially, video delivery has emerged as the main driving factor to lead a steep increase in wireless traffic in modern networks. As an example of deploying caches in wireless networks, authors in \cite{FemtoCaching} propose a system where helpers with low-rate backhaul but high storage capacity store popular video files. In the delivery phase the available files at the helpers are locally delivered to requesting users, and the main base station only transmits files not available from helpers, which relieves the burden on network backhaul links.

Along another research line, the combined use of wireless edge caching and coded multicasting has been suggested to concurrently
serve multiple unicast demands via coded multicast transmissions. This proposal, also known as \emph{Coded Caching}, is a breakthrough idea in this paradigm which was first proposed in \cite{Maddah_2014}. In a few words, in this method the caches are first filled in network low traffic hours without knowing the actual demands of users. In peak traffic hours, after revealing users' requests, coded messages are multicast to groups of users, which complete the delivery of distinct contents of users. Following this paradigm, decentralized coded caching scheme based on independent random content placement was introduced by Maddah-Ali and Niesen \cite{Maddah_2015}.  It was shown that, in the large file size regime, the multicasting gain is almost the same as the centralized caching scenario for large networks. The finite file size regime was considered in \cite{Shanmugam-2016}. Also, in \cite{S. Jin-2016} the authors proposed a caching scheme that outperforms the scheme in \cite{Maddah_2015} when the file size is not large.

In order to boost the performance of coded caching, many works have proposed using multiple transmitters/antennas, such as \cite{Shariatpanahi_2016}, \cite{Naderializdeh_2017}, and \cite{Lampiris_2018}. The main design in such works is to benefit from a combined multi-antenna multiplexing/global coded caching gain, by carefully borrowing ideas from the Zero-Forcing (ZF) techniques. It is shown that in such a scenario, one can have an additive gain from both coded caching and multi-antenna multiplexing domains. Further research have explored finite $\mathrm{SNR}$ performance analysis of multi-antenna coded caching in \cite{Shariatpanahi_2019}, and \cite{Tolli_2020}, which suggest similar gains.

	In this paper, we investigate a coded caching scenario with multiple transmitters (such as above works), but with a decentralized cache content placement phase assumption. The decentralized setting has been first investigated in a single transmitter setup in \cite{Maddah_2015} (and many follow up works), which  we extend to the multi-transmitter setup. In order to do this, we first derive the so called \emph{Coding Delay} of a content delivery algorithm which combines the decentralized content delivery algorithm in \cite{Maddah_2015} and the ZF-based delivery algorithm in \cite{Shariatpanahi_2016}, and discuss the insights it provides. Then, we consider a heterogeneous setup in which there are two groups of users at the placement phase. The first group can coordinate their caching phase and follow a centralized coded caching setup, while the users in the second group follow the fully decentralized caching scenario. At the delivery phase the groups are merged and the transmitters should fulfill all the demands. For this setup we analyze two delivery strategies. The first strategy treats the users in a TDMA delivery  fashion, while the second strategy treats them altogether. Furthermore, we investigate the effect of limited subpacketization in such scenarios and see how fast the performance converges to the infinite file size case, as a function of allowed file size.
	
	The most relevant work to ours is \cite{Thomdapuand_2019} which considers a similar problem in the wireless context, but with focus on finite SNR analysis and beamformer design for the multi-antenna transmitter. Since here we consider Coding Delay (which is equivalent to high $\mathrm{SNR}$ analysis of the wireless setup) we arrive at a much simpler closed-form performance expression which helps to clearly observe the role of multiple transmitters in the decentralized setting. Furthermore, in this work we consider a hybrid cache content placement scenario which has not been considered in  \cite{Thomdapuand_2019}. Moreover, here we analyze the finite file size regime as well.
	
	Finally, let us review some notations used in this paper. We use lower case bold-face symbols to represent vertical vectors, upper case bold-face symbols to represent matrices, and \emph{mathcal} symbols to show sets. For any matrix $\mathbf{A}$ (vector $\mathbf{a}$), matrix ${\mathbf{A}^{t}}$ (vector ${\mathbf{a}^{t}}$) denotes the transpose of matrix $\mathbf{A}$ (vector $\mathbf{a}$). For vector $\mathbf{a}$, ${\mathbf{a}^{\bot }}$ illustrates the condition ${\mathbf{a}^{t}}{\mathbf{a}^{\bot }}={{({\mathbf{a}^{\bot }})}^{t}}\mathbf{a}=0$ is satisfied. For any two sets $\mathcal{B}_{1}$ and $\mathcal{B}_{2}$ the set $\mathcal{B}_{1}\backslash \mathcal{B}_{2}$ contains those elements of $\mathcal{B}_{1}$ absent in $\mathcal{B}_{2}$. Also, we define $[{K}]=\{{1},...,{K}\}$ and $\mathbb{N}$ to be the set of natural numbers. Moreover, ${{\mathbb{F}}_{q}}$ shows a finite field with $q$ elements, and $\mathbb{F}_{q}^{a\times b}$ the set of all ${a}\tiny{-}{by}\tiny{-}{b}$ matrices whose elements belong to ${{\mathbb{F}}_{q}}$. In addition, let ${{y}_{1}},...,{{y}_{m}}\in {{\mathbb{F}}_{q}}$, then $\varphi{({y}_{1}},...,{{y}_{m})}$ is a random linear combination of ${{y}_{1}},...,{{y}_{m}}$ where the random coefficients are uniformly chosen from ${{\mathbb{F}}_{q}}$, and $\oplus$ represents addition in the corresponding finite field.

	\section{System Model and Assumption}
	In this paper, we consider a content delivery scenario where $L$ transmitters are connected to $K$ users via a \emph{Linear Network}, as shown in Figure \ref{fig:System-Model}. All the transmitters have full access to a contents library of $N$ files, i.e., $\mathcal{W}=\{W_1,\ldots,W_N\}$, in which the size of each file is $F$ bits. Also, each user is equipped with a cache of size $MF$ bits. We represent data by $m$-bit symbols as members of a finite field $\mathbb{F}_q$, where $q=2^m$.
	
	\begin{figure}[h]
		\includegraphics[width=8cm]{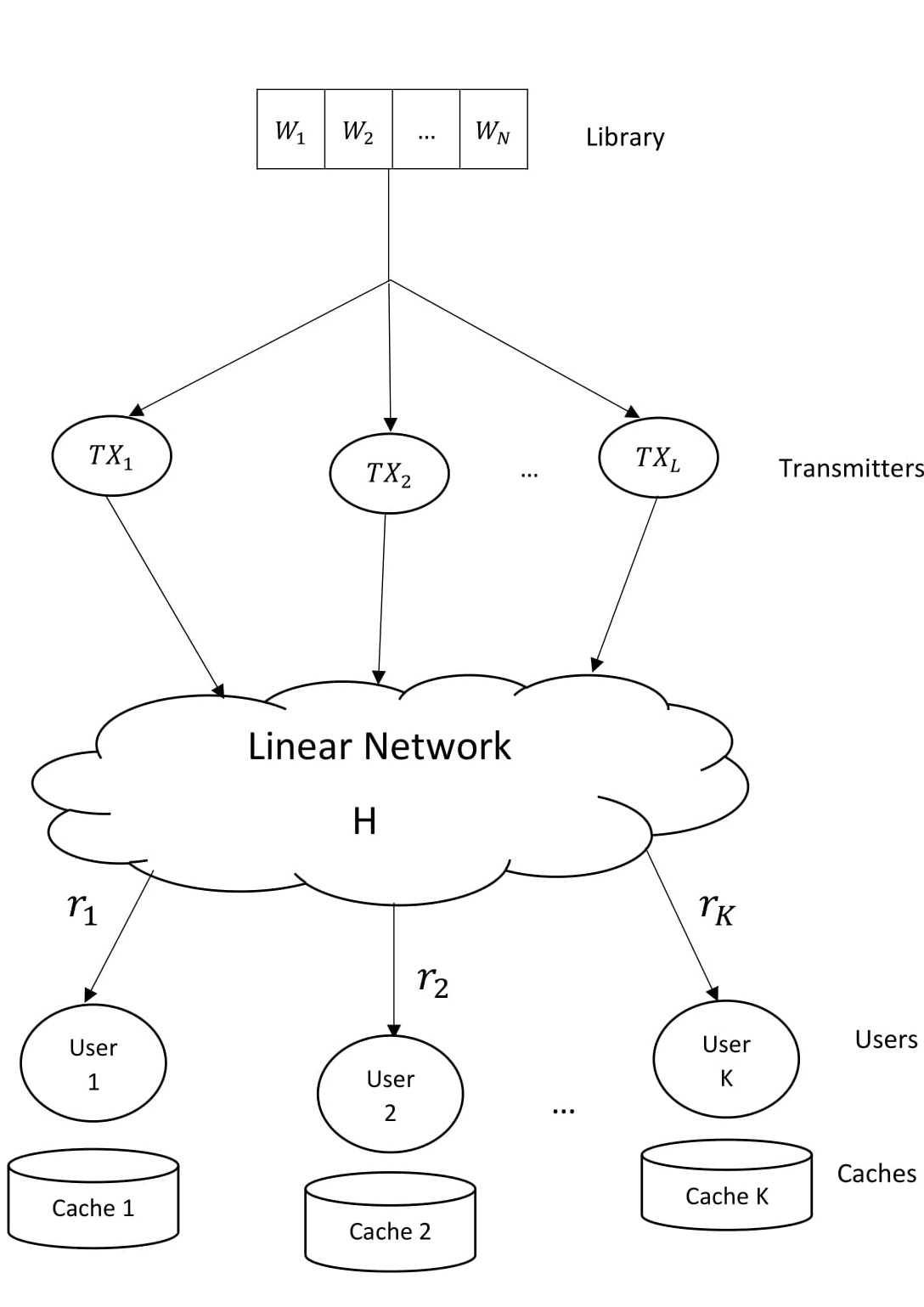}
		\caption{System Model: $L$ transmitters are connected to $K$ users via a linear network with transfer matrix $\mathbf{H}$.
			\label{fig:System-Model}}
	\end{figure}
	
	The system experiences two different traffic conditions, namely low-peak and high-peak hours, and thus operates in two phases. In the first phase, which we call the \emph{Cache Content Placement} phase and happens during network low-peak hours, each user caches $MF$ bits of data, which can be any function of the library $\mathcal{W}$,  without prior knowledge of their actual requests. If the users cooperate in the first phase we call it a centralized coded caching scenario, otherwise it is called a decentralized coded caching setup.
	
	In the second phase, which we call the \emph{Content Delivery} phase and happens during network high-peak hours, each user reveals a single content request, where the request of $i$-th user is denoted by $W_{d_i}$, and therefore the requested contents can be represented by the vector $\mathbf{d}=\{d_1,...,d_K\}$. In order to fulfill users' demands, the transmitters employ the linear network to deliver remaining portions of files not cached at the users. For the sake of presentation simplicity we assume $N>K$, while, extending the results to the other cases is straightforward.
	
	We assume an slotted time model for the delivery phase, where each slot represents a single channel use. More specifically, at the time slot $t=1,2,\ldots$, transmitter $i$ sends the symbol $s_i(t) \in \mathbb{F}_q$, which can be an arbitrary function of $\mathcal{W}$, and depends on the users requests $\mathbf{d}$, where $i=1,\ldots,L$. Also, the linear network connecting transmitters and users is assumed to be an \emph{error-free} \emph{zero-delay} network, and thus after injecting data from the transmitters, user $k$ receives $r_k(t)  \in \mathbb{F}_q$ for all $k=1,\ldots,K$. Due to the linear property of the channel we have the following relation between the symbols sent by the transmitter, and the symbols received by the users
	\begin{align}
		\mathbf{r}(t)=\mathbf{H}\mathbf{s}(t),
	\end{align}
	where
	\begin{center}
		$\mathbf{s}(t)\triangleq\begin{pmatrix}
			{{s}_{1}}(t) \\ 
			\vdots  \\ 
			{{s}_{L}}(t) \\ 
		\end{pmatrix}$
		\hspace{3mm} ,\hspace{3mm}
		$\mathbf{r}(t)\triangleq\begin{pmatrix}
			{{r}_{1}}(t) \\ 
			\vdots  \\ 
			{{r}_{K}}(t) \\ 
		\end{pmatrix},$ \\	
	\end{center}
	where $\mathbf{H} \in \mathbb{F}_q^{K\times L}$ represents the linear transformation imposed by the channel, which is assumed to be static during the delivery phase. For simplicity, we assume that the elements of $\mathbf{H}$ are uniformly at random chosen from the field $\mathbb{F}_q$, and the field size is large enough to fulfill all independence requirements used throughout the paper.
	
	The Delivery phase's primary purpose is to reduce the number of time slots required to respond to user requests. Let us define the number of time slots (channel uses) required to satisfy a specific users' demand vector $\mathbf{d}$ by $T_C({\mathbf{d}})$. Then, the \emph{Coding Delay}, which is our main performance metric in this paper, is defined to be 
	\begin{align}
		T_C=\max_{\mathbf{d}} T_C({\mathbf{d}}).
	\end{align}
	
	It should be noted that the \emph{linear network} model considered in this paper is very general and accommodates wired and wireless content delivery scenarios. For example, assuming that the transmitters and receivers are connected via a wired network modeled by a Directed Acyclic Graph (DAG), and intermediate nodes perform a low complexity topology-oblivious Random Linear Network Coding (RLNC) scheme, wired content delivery scenarios can be translated to this linear network model (see e.g., \cite{Shariatpanahi_2016}). On the other hand, if we assume a wireless Multiple-Input Single-Output Broadcast (MISO-BC) setup, the high SNR Degrees-of-Freedom (DoF) analysis can be translated to the coding delay analysis used in this paper (see e.g., \cite{Shariatpanahi_2019}).
	
	\section{Decentralized Multi – Transmitter Coded Caching}
	
	This section will show how adding transmitters can reduce the coding delay defined in the previous section in the decentralized coded caching setup. To do this, we adapt the zero-forcing based transmission scheme introduced in \cite{Shariatpanahi_2016} to the decentralized setup introduced by \cite{Maddah_2015} and analyze its performance in terms of coding delay.  The introduced gain by the proposed algorithm is achieved by taking advantage of the possibility of vector, and parallel transmissions for a decentralized coded caching system, known as a decentralized multi-transmitter coded caching algorithm.
	
	The primary approach is presented in Algorithm \ref{alg:DMCC} which provides the details of the cache content placement and content delivery phases.

	In Algorithm \ref{alg:DMCC}, $W_{{{d}_{r}},\mathcal{U}\backslash \{r\}}^{{{c}_{u}}}$ represents the piece of required file for user $r$ which is stored in the cache of all members of $\mathcal{U}$ except user $r$, and the $\varphi $  operator presents the linear combinations with random coefficients from the file components. Finally, ${{G}_{\omega }}(\mathcal{U})$ represents the linear combination of the file components whose coefficients are generated by the $\varphi$ operator.
	
In Theorem \ref{Theorem_Performance} we derive the performance of Algorithm \ref{alg:DMCC} in terms of coding delay, which clearly illustrates the benefits of using multiple transmitters in this setting.
	\begin{thm}\label{Theorem_Performance}
		For a decentralized multi-transmitter coded caching setup,  Algorithm \ref{alg:DMCC} achieves the following coding delay :
		\begin{center}
			\begin{align}
				\nonumber T_C&=\\
				&\sum_{\alpha =1}^{K}\sum_{ \tau \subseteq [K] , \lvert \tau\rvert =1}^{\dbinom{k}{\min(\alpha +L -1 , k }}  
				\dfrac{\begin{pmatrix}
						k-1 \\ 
						\alpha -1 \\ 
				\end{pmatrix}}{\binom{K-\alpha}{\min(\alpha +L-1,K) -\alpha }}X_{\alpha , \tau } ,
				\label{Eq_Main_CodingDelay_finite}
			\end{align}
		\end{center}
where $ X_{\alpha , \tau } $ is a random variable  with the distribution which we approximate by the Gamma distribution as follows: \\
	\begin{align}
		f(x;k;\theta ) = 
		{{x}^{k-1}}\frac{{{e}^{{}^{-x}/{}_{\theta }}}}{{{\theta }^{k}}\Gamma (k)}   \hspace{3mm}& \mbox{for }\hspace{1.5mm} x , k , \theta > 0 .
	\end{align}
Here, $\theta$ is a parameter estimated numerically.
	\end{thm}

	\begin{proof}
		Please refer to the Appendix in Section \ref{Append} .
	\end{proof}

Theorem \ref{Theorem_Performance} states the performance of the system in the finite file size regime, which results in a random coding delay, highlighted by the random variables $X_{\alpha , \tau }$. Following the concentration of measure approach used in \cite{Maddah_2015}, one can arrive at a much simpler result for the large file size regime, formally stated in the following corollary.
	\begin{corollary}
		\label{Theorem_Performance_corollary} In the limit of $F\rightarrow \infty$, Algorithm \ref{alg:DMCC} achieves the following coding delay (with high probability):\\
		\begin{align}
			\nonumber  T_C&=\\
			&\sum_{\alpha =1}^{K}{{K \choose \alpha}{{\left(\frac{M}{N}\right)}^{\alpha -1}}{{\left(1-\frac{M}{N}\right)}^{K-\alpha +1}}\frac{\alpha }{\min(\alpha +L-1,K)}}.
			\label{Eq_Main_CodingDelay_infinite}
		\end{align}
	\end{corollary}
	\begin{proof}
		Please refer to the Appendix in Section \ref{Append} .
	\end{proof}
	\begin{algorithm}[ht]
		
		\caption{Decentralized Multi-Transmitter Coded Caching }\label{alg:DMCC}
		
		\begin{algorithmic}[1]
			\STATE \textbf{Procedure PLACEMENT $(W_1 , W_2 , ... , W_N )$ }
			\FOR{  $ k \in [ K ] , n \in [ N ] $}
			\STATE user $k$ independently caches a subset of $ \frac{M}{N} F $ bits of file $ n $, chosen uniformly at random.
			\ENDFOR
			\STATE \textbf{end Procedore}
			
			\STATE \textbf{Procedure DELIVERY $(d_1 , d_2 , ... , d_K )$ }
			\STATE $ x \leftarrow 0 $
			\FOR{  $   \alpha  = K , K-1 , ... , 1  $}
			\STATE $ c_u  \leftarrow 1 $
			\FOR{  $ \mathcal{T} \subseteq [ K ] ,  \lvert \mathcal{T} \rvert = \min \{ \alpha + L - 1 , K \}$ }
			\FOR{$\omega = 1 , 2 , ... , \binom{ \min{(\alpha + L - 1,K)}-1 }{ \alpha - 1 } $}
			\STATE $ x \leftarrow x+1 $
			\FOR{$ \mathcal{U} \subseteq [ \mathcal{T} ] \hspace{3mm},\hspace{3mm} \lvert \mathcal{U} \rvert = \alpha  $}
			\STATE Design $ \boldsymbol{\psi} ^{ \mathcal{T}\setminus \mathcal{U} }  $ for all $ i \in \mathcal{T} \setminus \mathcal{U} :  \boldsymbol{h_i} \perp \boldsymbol{\psi}^{ \mathcal{T}\setminus \mathcal{U} } $
			\STATE $ G_{\omega} ( \mathcal{U} ) \leftarrow \varphi^{(\omega)}_{ r \in \mathcal{U} } \left( W^{c_u}_{d_{r} , \mathcal{U} \setminus \{r\}} \right) $
			\ENDFOR
			\STATE Server send  $ Y_x \leftarrow \sum_{\mathcal{U} \subseteq [ \mathcal{T} ] } G_{\omega} ( \mathcal{U} )  \boldsymbol{\psi}^{ \mathcal{T}\setminus \mathcal{U} } $
			\ENDFOR
			\STATE $ c_{u} \leftarrow c_{u} +1 $
			\ENDFOR
			\ENDFOR
			\STATE \textbf{end Procedore}

		\end{algorithmic}
		
	\end{algorithm}
	
	The formal proof of Theorem \ref{Theorem_Performance} and Corollary \ref{Theorem_Performance_corollary} is provided in the Appendix, but here we first discuss the implications of the results and then go through the main concepts behind the corresponding algorithm.  The only difference of \eqref{Eq_Main_CodingDelay_infinite} and the result of \cite{Maddah_2015} is that here in each term of the summation, we have a spatial multiplexing based delay reduction factor of $\frac{\alpha }{\min(\alpha +L-1,K)}$ which accounts for the benefits of adding multiple transmitters to increase the multicast group size for each transmission from $\alpha$ to $\min(\alpha +L-1,K)$. Therefore, it can easily be noted that when $L=1$, the coding delay in \eqref{Eq_Main_CodingDelay_infinite} reduces to the one in \cite{Maddah_2015}.

	The proposed algorithm consists of two main phases, in which the first phase is identical to the cache content placement introduced in \cite{Maddah_2015}, where each user randomly and independently of other users caches the amount of $\frac{M}{N}F$ bits of each file in their memory. In the delivery phase, we adapt the delivery algorithm proposed in \cite{Maddah_2015} to the multi-transmitter setup. Therefore, to respond to all user requests, we consider the parameter $\alpha$, such that for each $\alpha$ a common coded message can be helpful for $\min(\alpha+L-1,K)$ users, in which the global coded caching gain and multi-transmitter spatial multiplexing gain is included. The parameter $\alpha$ spans all possible values to benefit from all coded multicasting opportunities.

	In order to gain more insight into the above result, one can derive the following lower bound on $T_C$:
	\begin{align}
		\nonumber  T_C (L)\geq T_C(L=1)-\Delta T_C(L),
	\end{align}
	where $T_C(L=1)$ is the coding delay when we have only a single transmitter (the same as the result in \cite{Maddah_2015}), and
	\begin{align}
		\nonumber &\Delta T_C(L)\triangleq \\ \nonumber
		&(L-1) \sum_{\alpha =1}^{K}{{K \choose \alpha}{{\left(\frac{M}{N}\right)}^{\alpha -1}}{{\left(1-\frac{M}{N}\right)}^{K-\alpha +1}}\frac{1 }{\alpha +L-1}}
	\end{align}
	is the coding delay reduction factor due to deploying multiple transmitters (i.e., multiplexing gain).

	
	\section{Hybrid Cache Content Placement -- Limited Cooperation Setup}
	The main distinction between centralized and decentralized coded caching schemes is at the cache content placement phase. In the centralized setting, we assume full cooperation between the users, while in the decentralized setting no cooperation is assumed. In this section we propose an in-between setting which we call limited cooperation setup. Under this setup we consider $K$ users partitioned into two groups. The first group (group A) consists of $K_c$ users which can fully cooperate at the cache content placement phase, while the second group (group B) of $K_d=K-K_c$, do not cooperate at all and use completely random content placement strategy. In the delivery phase, the two groups merge together and request $K$ files from the server, and the total time needed for the server to satisfy the requests of all $K$ users is the corresponding coding delay.
	
As the baseline delivery strategy one can consider serving the two groups one after another, which we call the Hybrid-TDMA strategy. The total coding delay of this strategy will be
\begin{align}
		T_{TDMA}&= \\ \nonumber
			&\frac{{K}_{c}(1-\frac{M}{N})}{L+\frac{{K}_{c}M}{N}} \\ \nonumber
			+\sum\limits_{\alpha =1}^{{K}_{d}} & {K_d \choose \alpha} {\left(\frac{M}{N}\right)^{\alpha -1}} (1-\frac{M}{N})^{K_d-\alpha +1} \frac{\alpha }{\min(\alpha +L-1,K_d)}
\end{align}
In order to get more insight into the problem let us consider the single transmitter scenario. For the case of $L=1$, the above equation reduces to:
	\begin{align}
		\label{eq.hybrid_L=1}
		\frac{{K}_{c}\left(1-\frac{M}{N}\right)}{1+\frac{{K}_{c}M}{N}}+{K}_{d}\left(1-\frac{M}{N}\right) \frac{N}{{K}_{d}M}\left(1-{{\left(1-\frac{M}{N}\right)}^{{K}_{d}}}\right)
	\end{align}
By focusing on the low memory regime and using the Taylor series approach we can derive
	\begin{align}
		T_{Hybrid} \approx K-\frac{M}{N}\left(K_{c}^{2}+K+\begin{pmatrix}
			{{K}_{d}} \\ 
			2 \\ 
		\end{pmatrix} \right)
	\end{align}

	\begin{align}
		T_{Centralized} \approx K-\frac{M}{N} \left({{K}^{2}}+K \right)
	\end{align}
	
	\begin{align}
		T_{Decentralized} \approx K-\frac{M}{N}\left(K+ \begin{pmatrix}
			K \\ 
			2 \\ 
		\end{pmatrix} \right)
	\end{align}	
which clearly illustrate the caching gain as the reduction in the coding delay in all three scenarios, namely hybrid, centralized and decentralized settings. First, one can note that by setting $K_c=K$ and $K_d=K$, $T_{Hybrid}$ reduces to $T_{Centralized}$ and $T_{Decentralized}$, respectively. The main question is if the hybrid setting is always superior to the decentralized setting or not. One can see that this happens only when
\[K_{c}^{2}+ \begin{pmatrix}
		{{K}_{d}} \\ 
		2 \\ 
	\end{pmatrix} >  \begin{pmatrix}
		K \\ 
		2 \\ 
	\end{pmatrix}\]
is satisfied.

As an alternative approach to the baseline TDMA scheme discussed above, we use Algorithm  \ref{alg:DMCC} for the delivery phase of the Hybrid cache content placement scenario. The performance of this scheme is explored in the next section.

\section{Simulation Result}
In this section, we first illustrate the role of using multiple transmitters in reducing the coding delay in the decentralized, setting, then we investigate the proposed hybrid caching scenario, and finally justify our choice of Gamma distribution in Theorem \ref{Theorem_Performance}.

	\subsection{The Effect of the number of transmitters on the performance of decentralized coded caching systems}
	Figure \ref{fig:Effect-L-Finite-difrent M} shows the effect of the number of transmitters on the performance of decentralized coded caching systems in the case of finite and infinite file-size for different values of $M$. This figure clearly shows the benefits of the multiplexing gain provided by multiple transmitters in the decentralized setting, similar to the centralized setting. Also, this figure shows the price to pay for the limited subpacketization at the finite file size regime, in terms of coding delay. 

On the other hand, Figure \ref{fig:Effect-L-Finite-difrent F} plots the coding delay as a function of file size which clearly illustrates the convergence to the infinite file size result as $F$ increases.
	
	\begin{figure}[h]
		\includegraphics[width=9cm]{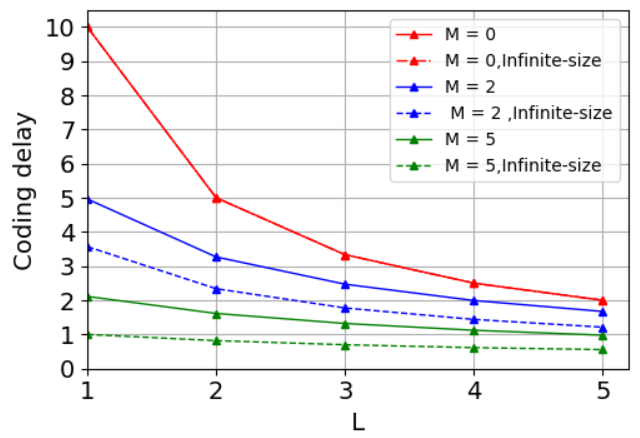}
		\caption{Investigating the effect number of transmitters on the performance of decentralized coded caching System with $K = 10$ users and $N = 10 $ files and different $M$ values and $F = 100 $}
		\label{fig:Effect-L-Finite-difrent M}
	\end{figure}
	
		\begin{figure}[h]
		\includegraphics[width=8cm]{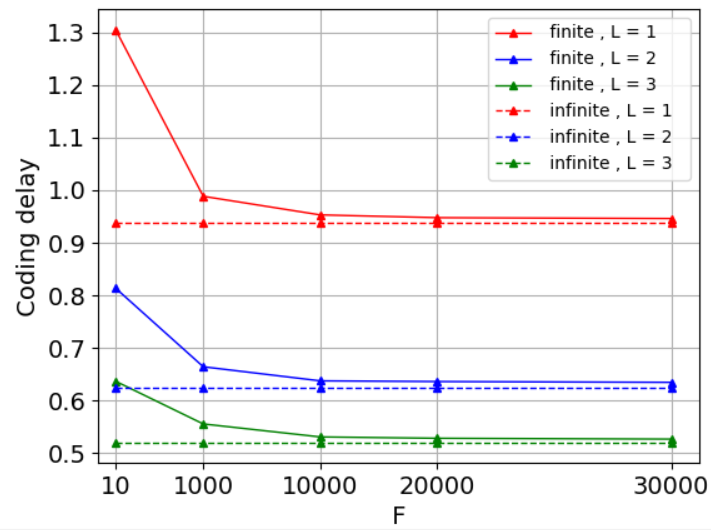}
		\caption{Investigating the effect of transmitters on the performance of decentralized coded caching System with $K =4$ users and $N = 4$ files and cache size $ M = 2 $ and different $F$ values}
		\label{fig:Effect-L-Finite-difrent F}
	\end{figure}

	\subsection{ Comparison between different caching strategies }
Here, we compare the coding delay for centralized, decentralized, and hybrid schemes for different parameters for the finite file size regime. In Fig. \ref{fig:hybrid1}, coding delay is plotted for these setups as a function of $M$, for $K=8$ users, and based on the TDMA scheme. For the hybrid scheme we assume two groups of $4$ users. As expected, all the three schemes' delays reduces by increasing $M$, and the centralized scheme has the lowest delay among all. However, interestingly, decentralized setup outperforms the hybrid scheme. In Fig. \ref{fig:hybrid2} we have the same parameters as in Fig. \ref{fig:hybrid1}, only changing the number of the second group in the hybrid scheme to $2$. Here we see that the hybrid scheme outperforms the decentralized scheme.	
	\begin{figure}[h]
		\includegraphics[width=8cm]{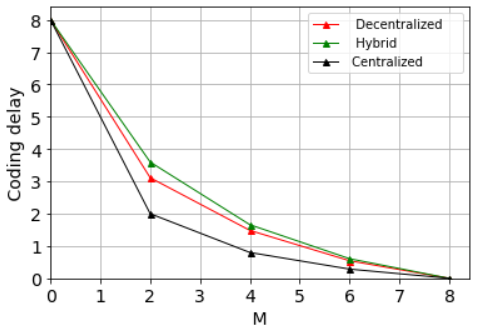}
		\caption{Comparison between the performance of the presented method in [3] and Algorithm 1 and Hybrid Strategy, for a system with $K = 8$ users $N = 8$ files and file size $F = 1000$, in this figure it is assumed that half of the users are members of group A and half of them are members of group B}
		\label{fig:hybrid1}
	\end{figure}
	
	\begin{figure}[h]
		\includegraphics[width=8cm]{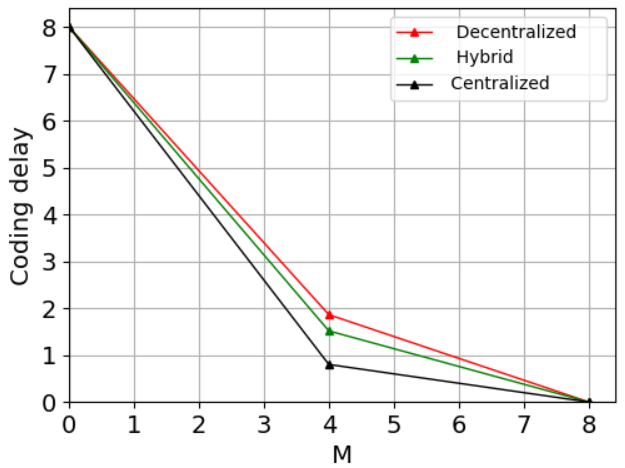}
		\caption{Comparison between the performance of the presented method in [3] and Algorithm 1 and Hybrid Strategy, for a system with $K = 8 $ users $N = 8 $ files and file-size $ F = 1000 $, in this figure it is assumed that about $75$ percent of users are members of group A and $25$ percent of them are members of group B}
		\label{fig:hybrid2}
	\end{figure}
	As Figures \ref{fig:hybrid1} and \ref{fig:hybrid2} show, our Hybrid Strategy method is somewhere between full centralized and full decentralized in terms of performance. Consequently, the hybrid strategy could perform better than full decentralized in some regimes. According to the simulation results, it is clear that the hybrid strategy's performance for caching systems depends on the number of users in groups A and B. So if more users are in group A, the hybrid strategy would have better performance, and its functionality will get closer to the complete centralized scheme.
	
	Finally, Figure \ref{fig:hybrid&pure}  plots coding delay of the hybrid scheme as a function of the number of users in the centralized group i.e., $K_c$ (for fixed $K=10$), which shows that the hybrid scheme outperforms the decentralized scheme only in certain regimes. In this figure, in addition to the baseline TDMA scheme, we have also included delivery coding delay based on Algorithm \ref{alg:DMCC} (with the label "Hybrid" in the figure).

	\begin{figure}[h]
		\includegraphics[width=9cm]{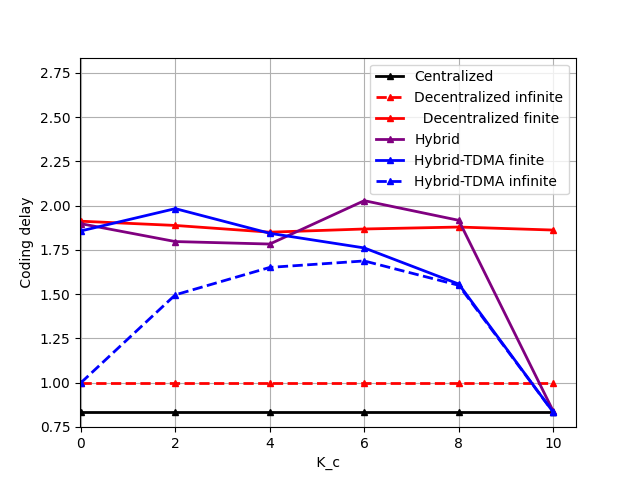}
		\caption{Comparison between hybrid  and pure strategies , for a system with $K = 10 $ users $N = 10 $ files, cache size $ M = 5 $ and file size $F = 1000$}
		\label{fig:hybrid&pure}
	\end{figure}

	\subsection{Fitting distribution to the data}
	\begin{figure}[htb]
		\includegraphics[width=8cm]{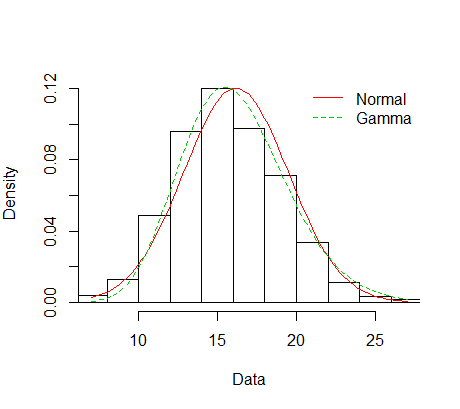}
		\caption{Fitting Gamma distribution to the data}
		\label{fig:fitData}
	\end{figure}
	
	As we said in Theorem \ref{Theorem_Performance}, the coding delay in the finite file size regime is a random variable and is a function of some other random variables named $ X_{\alpha , \tau } $. Here we justify our choice of Gamma distribution to approximate the distribution of $ X_{\alpha , \tau }$ by numerical data fit approach in Figure \ref{fig:fitData}.

	\section{Conclusion}
	In this paper we have considered downlink transmission in a multi-transmitter linear network where transmitters are connected to multiple cache-enabled end-users. We have introduced a new scheme that can be used for any type of caching strategy (centralized, decentralized, and hybrid). We examined the performance of this scheme specifically for the decentralized multi-transmitter coded caching systems by coding delay as the main metric. We also investigated a new practical scheme called hybrid strategy, which stands somewhere between full centralized and full decentralized. From simulation results, we have illustrated that the functionality of the decentralized scheme improves by adding transmitters, and the hybrid strategy performs better than the full decentralized algorithm, in certain regimes. For future work we look for more complex hybrid settings and investigating the performance of them.

	\section{Appendix \label{Append}}
	
	\subsection{Proof of Theorem \ref{Theorem_Performance}}
	
	\begin{proof}
		
		By analyzing the performance of Algorithm \ref{alg:DMCC}, we obtained the value of coding delay for a linear network with $L$ transmitters and $K$ users, and also we proved that the value of the coding delay obtained from the analysis of Algorithm \ref{alg:DMCC} is equal to the value expressed in Theorem \ref{Theorem_Performance}.\\
		Algorithm \ref{alg:DMCC} consists of two steps, placement and delivery. The first step is the same as \cite{Maddah_2015}, each user randomly and independently of other users caches  $\frac{M}{N}$F bits of each file in its memory. So the total memory size of each user is equal to : 
		\begin{center}
			$ \dfrac{M}{N} F\times N=MF$\\
		\end{center}
		
		Delivery phase: To respond to all users' requests, i.e.$({{d}_{1}},{{d}_{2}},...,{{d}_{k}})$, we consider the parameter $\alpha$, which refers to a certain number of users each time and can take values from $K$ to $1$, respectively. When $K$ = $\alpha$, the corresponding transmission is full multicast, and it is no longer possible to improve the coding delay. The related transmission is specified differently from the other ones, specified in line 11 of the Algorithm \ref{alg:DMCC}.
		
		For other values of $\alpha$, the arbitrary subset $\mathcal{T}$ from the users set is considered, the size of which is specified as follows:
		\begin{align}
			\left| \mathcal{T} \right|=\min(\alpha +L-1,K)
		\end{align}
		For a specific $\mathcal{T}$, all its $\alpha$ -member subsets are specified as follows:
		\begin{align}
			{\mathcal{U}_{i}}=\begin{pmatrix}
				\alpha +L-1 \\ 
				\alpha  \\ 
			\end{pmatrix} \hspace{3mm} , \hspace{2mm}i=1,2,..., \\
			\nonumber
			\vspace{3mm}
		\end{align}
		\begin{align}
			\nonumber
			( \hspace{2mm}\mathcal{U}\subseteq \mathcal{T},\left| \mathcal{U} \right|=\alpha \hspace{2mm})\\
		\end{align}
		For each given set of $\mathcal{T}$ and the corresponding $\mathcal{U}_{i}$ we design $\boldsymbol{h}^{\mathcal{T}\backslash {\mathcal{U}_{i}}}$ the vector in such a way that :
		\begin{align}
			{\mathbf{h}_{K}} \hspace{2mm}\bot\hspace{2mm} \boldsymbol{h}^{\mathcal{T}\backslash {\mathcal{U}_{i}}} \hspace{2mm} if \hspace{2mm}  K\in \mathcal{T}\backslash {\mathcal{U}_{i}}
			\label{eq.h-k}
		\end{align}
		Then we define for each $\mathcal{U}_{i}$:
		\begin{align}
			G(\mathcal{U})={{\varphi }_{r\in \mathcal{U}}}(W_{{{d}_{r}},\mathcal{U}\backslash \{r\}}^{{{c}_{u}}})
			\label{eq.G(U)}
		\end{align}
		
		\begin{align}
			{{c}_{u}}=\begin{pmatrix}
				K-\left| \mathcal{U} \right| \\ 
				\left| \mathcal{T} \right|-\left| \mathcal{U} \right| \\ 
			\end{pmatrix}
		\end{align}
		$W_{{{d}_{r}},\mathcal{U}\backslash \{r\}}^{{{c}_{u}}}$ represents the piece of required file for user $r$ which is cached in the memory of all members of $\mathcal{U}$ except user $r$, also ${c}_{u}$ shows that each subfile is divided into several mini files. In this equation, the $\varphi $ operator represents linear combinations with random coefficients from the file components. As a result, the transmitted vector corresponding to the set $\mathcal{T}$ is defined using equation \eqref{eq.h-k} and \eqref{eq.G(U)} as follows:
		
		\begin{align}
			Y=\sum\nolimits_{\mathcal{U}\subseteq \mathcal{T}}{G(\mathcal{U})\boldsymbol{h}^{\mathcal{T}\backslash \mathcal{U}}}\\
			\nonumber
		\end{align}
		This transmission will be useful for all members of $\mathcal{T}$. Each user can exploit the desired parts of the file from each transmission. For each set of $\mathcal{T}$, the corresponding transmission vector is repeated $\omega =\begin{pmatrix}
			\min (\alpha +L-1,K)-1 \\ 
			\alpha -1 \\ 
		\end{pmatrix}$ times with different linear coefficients.\\
		Thus G represents linear combinations of file fragments with random coefficients, which are likely to be independent of each other, the final equation of transmissions corresponding to any subset of
		$ \mathcal{T} \hspace{2mm},\left| \mathcal{T} \right|=\min(\alpha +L-1,K)$  is obtained as :\\
		
		\begin{align}
			{{G}_{\omega }}(\mathcal{U})=\varphi _{r\in \mathcal{U}}^{\omega }(W_{{{d}_{r}},\mathcal{U}\backslash \{r\}}^{{{c}_{u}}})\\
			\vspace{3mm}
			\label{eq.coded_messege}
			\nonumber
		\end{align}
		
		\begin{align}
			{{Y}_{\omega }}(\mathcal{T})=\sum\nolimits_{\mathcal{U}\subseteq \mathcal{T}}{{{G}_{\omega }}(\mathcal{U})\boldsymbol{h}^{\mathcal{T}\backslash \mathcal{U}}}\\
			\nonumber
		\end{align}
		As a result, for a specific subset $\mathcal{T}$, transmitters will send the following block:
		\begin{align}
			[\hspace{2mm}{{Y}_{1}}(\mathcal{T})\hspace{1mm},\hspace{1mm}{{Y}_{2}}(\mathcal{T})\hspace{1mm},...\hspace{1mm},\hspace{1mm}{{Y}_{\omega }}(\mathcal{T})\hspace{2mm}]\\
			\vspace{5mm}
			\nonumber
		\end{align}
		\begin{align}
			\hspace{2mm} \omega=1,2,...,\begin{pmatrix}
				\min (\alpha +L-1,K)-1 \\ 
				\alpha -1 \\ 
			\end{pmatrix}\
		\end{align}
		In equation \eqref{eq.coded_messege}, $G$ represents the linear combination of the file components whose coefficients are generated by the $\varphi$ operator and benefited by the users of set $\mathcal{T}$. The transmission corresponding to the intended set $\mathcal{T}$ would be beneficial for all its users. By using these transmissions, members of the set $\mathcal{T}$ will be able to receive some of the requested parts of the file. After the above process is fully implemented for the specified subset $\mathcal{T}$ , we repeat the same process for all subsets of $\mathcal{T}$. Eventually, all users will receive their requested files completely.\\
		Regarding the provided proof for the correctness of the content delivery strategy of Algorithm \ref{alg:DMCC}, we now calculate the corresponding coding delay. For a specific subset $\mathcal{T}$, each transmitted vector ${Y}_{\omega}$ has the following size :
		\begin{align}
			\frac{{{\Big(\dfrac{M}{N}\Big)}^{\alpha -1}}{{\Big(1-\dfrac{M}{N}\Big)}^{K-\alpha +1}}}{\begin{pmatrix}
					K-\alpha  \\ 
					\min (\alpha +L-1,K)-\alpha  \\ 
			\end{pmatrix}}
		\end{align}
		, which is the result of \cite{Maddah_2015}. So the sent block corresponding to the subset
		$ \mathcal{T} $
		has the size equal to  :
		\vspace{4mm}
		\begin{align}
		\begin{pmatrix}
			\min (\alpha +L-1,K)-1 \\ 
			\alpha -1 \\ 
		\end{pmatrix} \times \frac{{{\Big(\frac{M}{N}\Big)}^{\alpha -1}}{{\Big(1-\frac{M}{N}\Big)}^{K-\alpha +1}}}{ \begin{pmatrix}
				K-\alpha  \\ 
				\min (\alpha +L-1,K)-\alpha  \\ 
		\end{pmatrix}}
		\hspace{1cm}
		\end{align}

		\vspace{4mm}
		and since each subset $\mathcal{T}$, $\left| \mathcal{T} \right|= \min(\alpha +L-1,K)$  also has such a sent size, so the total amount of coding delay will be equal to :\\
		\begin{center}
			\begin{align}
				\label{eq.T(M).finite}
				T(M)=\sum\nolimits_{\alpha =1}^{K}\sum\nolimits_{\lvert \tau\rvert =1}^{z} 
				\dfrac{\begin{pmatrix}
						\min (\alpha +L-1,K)-1 \\ 
						\alpha -1 \\ 
				\end{pmatrix}}{\binom{K-\alpha}{\min(\alpha +L-1,K) -\alpha }}X_{\alpha , \tau } =
			\end{align}
		\end{center}
		\begin{center}
			\begin{align}
				\nonumber
				\sum_{\alpha =1}^{K-L+1}\sum_{ \tau \subseteq [K] , \lvert \tau\rvert =1}^{\begin{pmatrix}
						K  \\ 
						\alpha +L -1 \\ 
				\end{pmatrix}} 
				\dfrac{\begin{pmatrix}
						K -1 \\ 
						\alpha -1 \\ 
				\end{pmatrix}}{\binom{K-1}{\alpha +L-2  }}X_{\alpha , \tau }\hspace{2mm}  + 
			\end{align}
		\end{center}
		\begin{center}
			\begin{align}
				\nonumber
				\sum_{\alpha =K-L+2}^{K}\sum_{ \tau \subseteq [K] , \lvert \tau\rvert =1}^{\begin{pmatrix}
						K  \\ 
						K \\ 
				\end{pmatrix}} 
				\begin{pmatrix}
					K -1 \\ 
					\alpha -1 \\ 
				\end{pmatrix}X_{\alpha , \tau }   
			\end{align}
		$z =\begin{pmatrix}
			K \\ 
			\min (\alpha +L-1,K) \\ 
		\end{pmatrix}$
		\end{center}
		\vspace{2mm}
		\subsection{Proof of Corrallary \ref{Theorem_Performance_corollary}}
		For infinite file size total coding delay will be equal to :\\
		\begin{align}
			\label{eq.T(M).Infinite}
			T(M)=\sum\nolimits_{\alpha =1}^{K}\begin{pmatrix}
				K \\ \min(\alpha +L-1,K)
			\end{pmatrix}\times
		\end{align}
		\vspace{2mm}
		\begin{align}
			\nonumber
			\dfrac{1}{\binom{K-\alpha}{\min(\alpha +L-1,K) -\alpha }} \times \Big(\frac{M}{N}\Big)^{\alpha -1} \Big( 1-\frac{M}{N} \Big) ^{K-\alpha +1}  
		\end{align}
		\vspace{2mm}
		\begin{center}
			$ \times \begin{pmatrix}
				\min (\alpha +L-1,K)-1 \\ 
				\alpha -1 \\ 
			\end{pmatrix}$
		\end{center}
	
	In the above equations,$\begin{pmatrix}
		K \\ 
		\min (\alpha +L-1,K) \\ 
	\end{pmatrix}$
	represents the number of $\mathcal{T}$ subsets from the users' sets. we show that equation \eqref{eq.T(M).Infinite} is equal to the value expressed in corrollary \ref{Theorem_Performance_corollary}.\\
	
	To simplify the proof process, we divide equation  \eqref{eq.T(M).Infinite} and \eqref{Eq_Main_CodingDelay_infinite} into two parts according to $\min (\alpha +L-1,K)$ .\\
	By rewriting equation \eqref{Eq_Main_CodingDelay_infinite} :\\
	
	\begin{center}
		$T_C= $ \\
		\vspace{3mm}
		$\sum\nolimits_{\alpha =1}^{K}\dbinom{K}{ \alpha  }	\Big(\dfrac{M}{N} \Big)^{\alpha - 1} \Big( 1 - \dfrac{M}{N} \Big)^{ K - \alpha + 1} $ \\ 
		\vspace{4mm}
		$\times \dfrac{\alpha}{\min \Big( \alpha + L - 1 , K \Big)} = $\\
		\vspace{4mm}
		$  \sum\nolimits_{\alpha =1}^{K-L + 1}\dbinom{K}{ \alpha }	\Big(\dfrac{M}{N} \Big)^{\alpha - 1} \Big( 1 - \dfrac{M}{N} \Big)^{ K - \alpha + 1}  $\\
		\vspace{4mm}
		$ \times \dfrac{\alpha}{\alpha + L - 1 } \hspace{4mm} + $ \\
		\vspace{4mm}
		$  \sum\nolimits_{\alpha = K - L  + 2 }^{ K }\dbinom{K}{\alpha}\Big(\dfrac{M}{N} \Big)^{\alpha - 1} \Big( 1 - \dfrac{M}{N} \Big)^{ K - \alpha + 1}  $\\
		\vspace{4mm}
		$ \times \dfrac{\alpha}{K } $ \\
	\end{center}
	
	By rewriting equation \eqref{eq.T(M).Infinite} :
	\begin{center}
		$T(M)= $ \\
		\vspace{4mm}
		$\sum\nolimits_{\alpha =1}^{K} \begin{pmatrix}
			K \\ 
			\min(\alpha + L - 1 , K ) \\ 
		\end{pmatrix} \dfrac{\Big(\dfrac{M}{N} \Big)^{\alpha - 1 } \Big( 1 - \dfrac{M}{N} \Big) ^{K - \alpha + 1}}{{\begin{pmatrix}
					K-\alpha  \\ 
					\min (\alpha +L-1,K)-\alpha  \\ 
		\end{pmatrix}}}$\\
		\vspace{4mm}
		$\times \begin{pmatrix}
			\min (\alpha +L-1,K)-1 \\ 
			\alpha -1 \\ 
		\end{pmatrix} =	$\\
		\vspace{4mm}
		$ \sum\nolimits_{\alpha =1}^{K-L+1} 
		\begin{pmatrix}
			K \\ 
			\alpha + L - 1 \\ 
		\end{pmatrix} \dfrac{\Big(\frac{M}{N}\Big)^{\alpha -1}\Big(1-\frac{M}{N}\Big)^{K-\alpha +1}}{\begin{pmatrix}
				K-\alpha  \\ 
				L-1 \\ 
		\end{pmatrix}} $ \\
		\vspace{4mm}
		$ \times \begin{pmatrix}
			\alpha +L-2 \\ 
			\alpha -1 \\ 
		\end{pmatrix} \hspace{3mm} + $ \\
		\vspace{4mm}
		$ \sum\nolimits_{\alpha =K-L+2}^{K} \hspace{3mm}
		\begin{pmatrix}
			K \\ 
			K \\ 
		\end{pmatrix} \dfrac{\Big(\frac{M}{N}\Big)^{\alpha -1}\Big(1-\frac{M}{N}\Big)^{K-\alpha +1}}{\begin{pmatrix}
				K-\alpha  \\ 
				K-\alpha  \\ 
		\end{pmatrix}} \times \begin{pmatrix}
			K-1 \\ 
			\alpha -1 \\ 
		\end{pmatrix}=$\\
		\vspace{4mm}
		\vspace{4mm}
		$ \sum\nolimits_{\alpha =1}^{K-L+1}
		\dfrac{K!\times \alpha}{(K-\alpha)! \times \alpha! \times (\alpha+L-1)}\times$\\
		\vspace{4mm}
		$\Big(\frac{M}{N}\Big)^{\alpha -1}\Big(1-\frac{M}{N}\Big)^{K-\alpha +1}+$\\
		\vspace{4mm}
		$ \sum\nolimits_{\alpha =K-L+2}^{K}
		\dfrac{K!\times \alpha}{(K-\alpha)! \times \alpha! \times (K-\alpha)!}\times$\\
		\vspace{4mm}
		$\Big(\frac{M}{N}\Big)^{\alpha -1}\Big(1-\frac{M}{N}\Big)^{K-\alpha +1}=$\\
		\vspace{4mm}
		$  \sum\nolimits_{\alpha =1}^{K-L + 1}\dbinom{K}{ \alpha }	\Big(\dfrac{M}{N} \Big)^{\alpha - 1} \Big( 1 - \dfrac{M}{N} \Big)^{ K - \alpha + 1}  $\\
		\vspace{4mm}
		$ \times \dfrac{\alpha}{\alpha + L - 1 } \hspace{4mm} + $ \\
		\vspace{4mm}
		$  \sum\nolimits_{\alpha = K - L  + 2 }^{ K }\dbinom{K}{\alpha}\Big(\dfrac{M}{N} \Big)^{\alpha - 1} \Big( 1 - \dfrac{M}{N} \Big)^{ K - \alpha + 1} \times \dfrac{\alpha}{K}=T_C  $\\
		
	\end{center}
	
	\vspace{4mm}

	Therefore, the equality of the two equations has been proved. \\
	
\end{proof}


\begin{thebibliography}{00}
		
		\bibitem{Bastug_2014} E. Bastug, M. Bennis, and M. Debbah, ``Living on the edge: The role of proactive caching in 5G wireless networks,'' IEEE
		Commun. Mag., vol. 52, no. 8, pp. 82–89, Aug. 2014.
		\bibitem{Golrezaei_2013}  N. Golrezaei, A. Molisch, A. Dimakis, and G. Caire, ``Femtocaching and device-to-device collaboration: A new architecture
		for wireless video distribution,'' IEEE Commun. Mag., vol. 51, no. 4, pp. 142–149, April 2013.
		\bibitem{FemtoCaching}N. Golrezaei, K. Shanmugam, A. G. Dimakis, A. F. Molisch, and G. Caire, “Femtocaching: Wireless video content delivery through distributed caching
		helpers,” in INFOCOM, March 2012, pp. 1107–1115.
		\bibitem{coding for caching}Y. Fadlallah et al., "Coding for Caching in 5G Networks", IEEE Commun. Mag., vol. 55, no. 2, pp. 106-13, Feb. 2017
		\bibitem{caching at wireless edge}D. Liu, B. Chen, C. Yang and A.F. Molisch, "Caching at the wireless edge: Design aspects challenges and future directions", IEEE Commmun. Mag., vol. 54, no. 9, pp. 22-28, Sep. 2016.
		\bibitem{Maddah_2014} M. A. Maddah-Ali and U. Niesen,“ Fundamental limits of caching ,” IEEE Trans. Inf. Theory, vol. 60 , no. 5, pp.2856-2867 , May 			2014.
		\bibitem{Shariatpanahi_2016} S. P. Shariatpanahi, S. A. Motahari, and B. H. Khalaj, ``Multi-server coded caching,'' IEEE Trans. Inf. Theory, vol. 62, no. 12, pp. 7253-7271, 2016.
		\bibitem{Shariatpanahi_2019} S. P. Shariatpanahi, G. Caire and B. Hossein Khalaj, "Physical-Layer Schemes for Wireless Coded Caching," in IEEE Transactions on Information Theory, vol. 65, no. 5, pp. 2792-2807, May 2019
		\bibitem{Tolli_2020} A. Tölli, S. P. Shariatpanahi, J. Kaleva and B. H. Khalaj, "Multi-Antenna Interference Management for Coded Caching," in IEEE Transactions on Wireless Communications, vol. 19, no. 3, pp. 2091-2106, March 2020
		\bibitem{Naderializdeh_2017} N. Naderializadeh, M. A. Maddah-Ali and A. S. Avestimehr, ``Fundamental Limits of Cache-Aided Interference Management,'' in IEEE Trans. Inf. Theory, vol. 63, no. 5, pp. 3092-3107, 2017.
		\bibitem{Lampiris_2018}E. Lampiris and P. Elia, "Adding Transmitters Dramatically Boosts Coded-Caching Gains for Finite File Sizes," in IEEE Journal on Selected Areas in Communications, vol. 36, no. 6, pp. 1176-1188, June 2018
		\bibitem{Maddah_2015} M. A. Maddah-Ali and U. Niesen, "Decentralized coded caching attains order-optimal memory-rate tradeoff," IEEE/ACM Transactions on Networking, vol. 23, no. 4, pp. 1029-1040, 2015.
		\bibitem{Thomdapuand_2019} S. T. Thomdapuand and K. Rajawat, ``Decentralized Multi-Antenna Coded Caching with Cyclic Exchanges,'' ICC 2019 - 2019 IEEE International Conference on Communications (ICC), 2019.
		
		\bibitem{Shanmugam-2016} K. Shanmugam, M. Ji, A. M. Tulino, J. Llorca and A. G. Dimakis., "Finite-Length Analysis of Caching-Aided Coded Multicasting," in IEEE Transactions on Information Theory, vol. 62, no. 10, pp. 5524-5537, Oct. 2016, doi: 10.1109/TIT.2016.2599110.
		
		\bibitem{S. Jin-2016}S. Jin, Y. Cui, H. Liu and G. Caire, "Order-Optimal Decentralized Coded Caching Schemes with Good Performance in Finite File Size Regime," 2016 IEEE Global Communications Conference (GLOBECOM), 2016, pp. 1-7, doi: 10.1109/GLOCOM.2016.7842115.
	\end{thebibliography}
\end{document}